\documentclass[11pt,a4paper]{article}

\usepackage[T1]{fontenc}
\usepackage[cp1250]{inputenc}

\usepackage{amssymb,amsthm,amsmath, amsfonts, stmaryrd} 
\usepackage{bm} 
\input xy
\xyoption{all}

\usepackage{hyperref} 

\newtheorem{thm}{Theorem}[section]

\theoremstyle{definition}

\newtheorem{remark}[thm]{Remark}

\numberwithin{equation}{section}

\newcommand{\X}{{\cal X}}
\newcommand{\R}{\mathbb{R}}
\newcommand{\eps}{\varepsilon}
\newcommand{\del}{\delta}
\newcommand{\g}{\mathfrak{g}}

\newcommand{\ra}{\rightarrow}
\newcommand{\lra}{\longrightarrow}
\newcommand{\rel}{\rightarrowtriangle}

\newcommand{\Sec}{\operatorname{Sec}}
\newcommand{\id}{\operatorname{id}}
\newcommand{\dd}{\mathrm{d}} 
\newcommand{\sT}{\textrm{T}}
\newcommand{\T}{\textrm{T}} 
\newcommand{\pa}{\partial} 
\newcommand{\p}{\operatorname{p}}

\newcommand{\wh}{\widehat}
\newcommand{\wt}{\widetilde}

\newcommand{\TT}{\mathcal {T}} 
\newcommand{\WW}{\mathcal {W}} 
\newcommand{\C}{\mathcal {C}} 
\newcommand{\PP}{\mathcal {P}} 

\newcommand{\LL}{\mathcal{L}}
\def\<#1>{\big\langle #1\big\rangle} 

\begin{document}

\title{Generalized Livens theorem and constrained Hamiltonian dynamics on algebroids}
\author{Micha\l \ J\'o\'zwikowski\\
 Institute of Mathematics, Polish Academy of Sciences\\\'Sniadeckich 8,
P.O. Box 21, 00-956 Warszawa, Poland\\{\tt mjoz@impan.gov.pl}
}
\maketitle
\begin{abstract}
We generalise Livens theorem, showing that Hamiltonian equation on the vector bundle $E^\ast\rightarrow M$, dual to a general algebroid $E\rightarrow M$, can be derived by means of a variational principle. The framework can be used to describe Hamiltonian dynamics associated with (both vaconomic and nonholonomic) constraints in the bundle $E$.
\bigskip

\bigskip

\end{abstract}

\section{Introduction}

\paragraph{Lagrangian mechanics on general algebroids}
The importance of a Lie algebroid structure in the formulation of Geometric Mechanics is now commonly accepted. The basic reason is that Lie algebroids cover, on the one hand, the structure of the tangent bundle of a manifold and, on the other hand, geometric structures which arise from a reduction by a symmetry group such as a Lie algebra of a Lie group or, more generally, a space $\T P/G$ of $G$-invariant vectors on a principal bundle $G\ra P\ra M$. Consequently, the formulation of Geometric Mechanics on Lie algebroids allows one to treat many natural situations in frames of a single universal formalism \cite{L,W}. 

Such a formulation was proposed by many authors \cite{CLMDM, GGU, LMM,IMDS,IMPS, M1, M2}. In this paper we will base on an approach by Grabowska, Grabowski and Urbański \cite{GGU,GG}, in which the dynamics associated with an algebroid $E$ is derrived by means of a certain double vector bundle morphism, naturally generalising the celbrated Tulczyjew triple \cite{Tul}. These authors were also the first to realise that neither the Jacobi identity, nor the skew-symmetry of the bracket is necessary to generate the dynamics. Consequently, their framework can be formulated also for general algebroids (cf. \cite{GU}). Such a generalization proved to be useful, for instance, in the context of nonholonomic constraints \cite{GLMM}. 

A natural problem is to incorporate the constraints into the framework of geometric mechanics on algebroids. This can be relatively easily done for a Lagrangian formalism. Namely as shown in \cite{GG} the generalised Euler-Lagrange equations can be derived by means of variational principle for $E$-valued curves. Variational calculus is here understood as the study of the differential $\dd S_L(\gamma)$ of an action $$S_L(\gamma):=\int_{t_0}^{t_1}L(\gamma(t))\dd t,$$ 
where $\gamma:[t_0,t_1]\ra E$ is an $E$-valued curve. For a variation $\del\gamma\in \T_\gamma E$ we simply have 
$$\<\del\gamma,\dd S_L(\gamma)>:=\int_{t_0}^{t_1}\<\del\gamma(t),\dd L(\gamma(t))>\dd t.$$
Our attention is restricted to certain curves $\gamma$, called \emph{admissible trajectories} and certain variations $\del \gamma$, called \emph{admissible variations}, which can be naturally constructed from an algebroid structure on $E$ (see section \ref{sec:var_calc} for details). Constraints can be now understood as restrictions in the set of admissible variations and/or trajectories. Constrained Lagrangian dynamics on algebroids can be now derived from a variational principle for these restricted sets \cite{GG}.

For the constrained version of the Hamiltonian dynamics there is so far no such an easy approach, even if an algebroid is simply a tangent bundle. Typically the constrained Hamiltonian equations are derived from the constrained Euler-Lagrange equations via coordinate calculations under an additional assumption of hyperregularity of the Lagrangian. This has been done by \cite{PT} for the case of vaconomic constraints and by \cite{BS, SM} for nonholonomic constraints. 

\paragraph{Livens theorem}
In this paper we propose another method of generating the constrained Hamiltonian dynamics on general algebroids. Our motivation comes from a classical Livens theorem stating that the unconstrained Hamiltonian dynamics on $\T^\ast M$ has a variational formulation. 
Namely, the Hamiltonian trajectories can be described as the critical points of the action
$$S_H(x(\cdot),p(\cdot)):=\int_{t_0}^{t_1}\<\dot x(t),p(t)>-H(x(t),p(t))\dd t,$$
with respect to variations $(\del x(\cdot ),\del p(\cdot))$ induced by homotopies in $\T^\ast M$.
(We use standard notation $(x^a)$ for coordinates on $M$, $(x^a,\dot x^b)$ and $(x^a,p_b)$ for natural coordinates induced on $\T M$ and $\T^\ast M$).   

\paragraph{Our results}
If we reformulate the above results in terms of the $\T M\oplus_M\T^\ast M$-trajectories $(x(\cdot),y(\cdot),p(\cdot))$, namely as a study of the action 
$$\wt S_H(x(\cdot),y(\cdot), p(\cdot)):=\int_{t_0}^{t_1}\<y(t),p(t)>-H(x(t),p(t))\dd t,$$
under an additional admissibility condition $y(t)=\dot x(t)$, Livens theorem can be easily generalized to the context of general algebroids as is done in theorem \ref{thm:var_ham}. To be more precise, given a Hamiltonian function $H:E^\ast\ra\R$, we can define a function $\LL_H:E\oplus_ME^\ast\ra\R$ by the formula
$$\LL_H(Y):=\<\p_E(Y),\p_{E^\ast}(Y)>-H(\p_{E^\ast}(Y)),$$
where $\p_E:E\oplus_ME^\ast\ra E$ and $\p_{E^\ast}:E\oplus_ME^\ast\ra E^\ast$ are natural projections. Now for a trajectory $\Gamma:[t_0,t_1]\ra E\oplus_ME^\ast$ we can define the action
$$ S_{\LL_H}(\Gamma):=\int_{t_0}^{t_1}\LL_H(\Gamma(t))\dd t.$$ We restrict our attention to trajectories $\Gamma$ such that $\gamma:=\p_E\circ\Gamma:[t_0,t_1]\ra E$ is an admissible trajectory in $E$ (in the sense of \cite{GG}) and variations $\del \Gamma:[t_0,t_1]\ra \T_\Gamma(E\oplus_ME^\ast)\approx \T_{\p_E(\Gamma)}E\oplus_{\T M}\T_{\p_{E^\ast}(\Gamma)}E^\ast$ such that $\p_{\T E}\circ\del \Gamma:[t_0,t_1]\ra \T_\gamma E$ is an admissible variation in $\T E$ (in the sense of \cite{GG}). Extremals $\Gamma$ of  the action $S_{\LL_H}(\cdot)$ under these restrictions give Hamiltonian dynamic on $E^\ast$ defined by a vector field  $\X_H:=\iota_{\dd H}\Pi_E$, where $\Pi_E$ is a contravariant 2-tensor on  $E^\ast$ equivalent to the presence of an algebroid structure on $E$. Precisely, the $E^\ast$-part $\p_{E^\ast}\circ \Gamma$ is a Hamiltonian trajectory, and the $E$-part $\p_E\circ\Gamma=\T^\ast\pi\circ \dd H$ can be understood as a momenta - velocity correspondence (Legendre map). 

\paragraph{Constrained Hamiltonian dynamics}
 As we mentioned above, the derivation of the constrained Lagrangian dynamics on an algebroid $E$ was possible by restricting the sets of admissible variations and trajectories. On the other hand, the (unconstrained) Hamiltonian dynamics is obtained via a variational approach which uses the unconstrained admissible variations and trajectories on $E$. A natural idea is hence to use the variational principle for the Hamiltonian dynamics for sets of admissible variations and trajectories restricted by the constraints. In consequence, we obtain a constrained version of the Hamiltonian dynamics on general algebroids. The results agree with the equations derived by \cite{PT,BS,SM}, although in contrast to these papers we start from an arbitrary Hamiltonian function, not assuming hyperregularity. 

\paragraph{Organization of the paper} After introducing general algebroids in section \ref{sec:alg}, we sketch an abstract and quite a general approach to variational calculus (section \ref{sec:var_calc}). In section \ref{sec:lag_ham_alg} we show its particular realization in the case of Lagrangian dynamics on a general algebroids (following the original paper \cite{GG}). In the same section we also define Hamiltonian dynamics on algebroids. Section \ref{sec:results} is devoted to the formulation and proof of theorem \ref{thm:var_ham} which generalizes Livens theorem. Finally in section \ref{sec:constraints} we calculate the constrained Hamiltonian dynamics  in the vaconomic and nonholonomic case and compare our approach with other authors.

\section{General algebroids}\label{sec:alg}

\paragraph{Local coordinates, notation}
Let $M$ be a smooth manifold and let $(x^a), \ a=1,\dots,n$, be a coordinate system in $M$. We denote by 
$\tau_M: \sT M \rightarrow M$ the tangent vector bundle and by $\pi_M \colon \sT^\ast M\rightarrow M$ the
cotangent vector bundle. We have the induced (adapted) coordinate systems $(x^a, {\dot x}^b)$ in $\sT M$ and
$(x^a, p_b)$ in $\sT^\ast M$.
        More generally, let $\tau: E \rightarrow M$ be a vector bundle and let $\pi: E^\ast \rightarrow M$ be the dual bundle.
  Let $(e_1,\dots,e_m)$  be a basis of local sections of $\tau:
E\rightarrow M$ and let $(e^{1}_*,\dots, e^{m}_*)$ be the dual basis of local sections of $\pi:
E^\ast\rightarrow M$. We have the induced coordinate systems:
    $(x^a, y^i),  y^i=\iota_{e^{i}_*}$ {in} $E$, and
    $(x^a, \xi_i), \xi_i = \iota_{e_i}$ {in} $E^\ast$.

Consequently, we have natural local coordinates
\begin{align*}
    &(x^a, y^i,{\dot x}^b, {\dot y}^j )  \text{ in $\sT E$},
  &&(x^a, \xi_i, {\dot x}^b, \dot{\xi}_j) \text{ in $\sT E^\ast$},\\
  &(x^a, y^i, p_b, \pi_j) \text{ in $\sT^\ast E$},
    &&(x^a, \xi_i, p_b, \varphi^j) \text{ in $\sT^\ast E^\ast $} .
\end{align*}
For a curve $\gamma(t)\sim x^a(t)$ in $M$ we will denote its tangent prolongation by $\T(\gamma(t))\sim(x^a(t),\dot x^b(t))$ in $\T M$. 

\paragraph{General algebroids} An \emph{algebroid structure} on a vector bundle $\tau:E\lra M$ is given by a bilinear bracket $[\cdot,\cdot]$ on the space $\Sec(E)$ of (local) sections of $\tau$, together with a pair of vector bundle morphisms $\rho,\sigma:E\lra \sT M$ (left and right anchor maps) such that the compatibility condition 
\begin{equation}\label{eqn:lieb_rule}
[g\cdot X,f\cdot Y]=gf[X,Y]+g\cdot\rho(X)(f)Y-f\cdot\sigma(Y)(g)X
\end{equation}
is satisfied for every $X,Y\in\Sec(E)$ and $f,g\in C^\infty(M)$. 

In literature one considers often certain subclasses of algebroids. Under additional assumption  that the bracket is skew-symmetric (hence $\rho=\sigma$) we speak of \emph{skew-algebroids}. A skew-algebroid whose anchor map is an algebroid morphism, i.e. 
\begin{equation}\label{eqn:ala}
\rho\left([X,Y]\right)=[\rho(X),\rho(Y)]_{\sT M}
\end{equation}
is called \emph{almost-Lie algebroids}. Almost-Lie algebroids  with the bracket  satisfying the Jacobi identity (in other words: if the pair ($\Sec(E)$,$[\cdot,\cdot]$) is a Lie algebra) are called \emph{Lie algebroids}.  

In the context of mechanics it is convenient to think about an algebroid over $M$ as a generalisation of the tangent bundle of $M$. An element $a\in E$ has an interpretation of a generalized velocity with actual velocity $v\in\sT M$ obtained by applying the left anchor map $v=\rho(a)$.

In local coordinates $(x^a,y^i)$, as introduced at the beginning of this section, the structure of an algebroid on $E$ can be described in terms of function $\rho^a_i(x)$, $\sigma^a_i(x)$ and $c^i_{jk}(x)$ on $M$ given by
$$\rho(e_i)=\rho^a_i(x)\pa_{x^a},\quad\sigma(e_i)=\sigma^a_i(x)\pa_{x^a}\quad \text{and} \quad [e_i,e_j]=c^k_{ij}(x)e_k.$$

\paragraph{Examples}
\subparagraph{The tangent bundle $\sT M$.}
Every tangent bundle $\tau_M:\T M\ra M$ of a manifold $M$ posses a natural structure of a Lie algebroid. We simply take the anchor map to be $\id_{\T M}:\T M\ra\T M$ and the standard Lie bracket of vector fields on $M$ as the algebroid bracket.
\subparagraph{The Lie algebra $\g$ of a Lie group $G$.}
A Lie algebra $\g$ understood as a vector bundle over a single point also posses the structure of a Lie algebroid. The anchor map is in this case trivial, whereas the algebroid bracket is simply the Lie bracket $[\cdot,\cdot]_\g$. 

\subparagraph{The Atiyah algebroid $\T P/G$.}
An example generalizing the above two arises from the study of $G$-invariant vector fields on the total space of a principal $G$-bundle $G\ra P\overset{\pi}\ra M$. Every such field can be canonically identified with a section of the quotient bundle $E=\T P/G\ra M$. Since a Lie bracket of two $G$-invariant vector fields is again $G$-invariant, $\Sec(E)$ is equipped with a natural Lie algebra structure $(\Sec(E),[\cdot,\cdot]_E)$. This, together with the natural anchor map $\rho:E=\T P/G\ra\T M$ defined as $\rho([X])=\pi_\ast X$, forms a Lie algebroid. It is known as the \emph{Atiyah algebroid}.

The above are examples of Lie algebroids. Natural examples of skew-algebroids appeared first in the context of nonholonomic constraints (cf. for instance \cite{GLMM}), while the definition of a general algeboid was introduced in \cite{GU}, where some natural examples are given.

The structure of an algebroid has several equivalent descriptions. We will give a few of them which are particularly useful in variational calculus -- compare \cite{GG,GGU}. 

\paragraph{Algebroids as linear (2,0)-tensors on $E^\ast$}

The algebroid structure on $E$ corresponds to a presence of a linear (2,0)-tensor field $\Pi$ on $E^\ast$. In local coordinates 
$$\Pi=c^k_{ij}(x)\xi_k\pa_{\xi_i}\otimes\pa_{\xi_j}+\rho^b_i(x)\pa_{\xi_i}\otimes \pa_{x^b}-\sigma^b_i(x)\pa_{x_b}\otimes\pa_{\xi_i}.$$
The linearity of $\Pi$ means that the natural morphism $\iota_\Pi:\T^\ast E^\ast\ra\T E^\ast$ induced by the contraction is a morphism of double vector bundles.
\begin{equation}
\xymatrix{
 & \sT^\ast E^\ast \ar[rrr]^{\iota_\Pi} \ar[dr]^{T^\ast\pi}
 \ar[ddl]_{\pi_{E^\ast}}
 & & & \sT E^\ast\ar[dr]^{\sT\pi}\ar[ddl]_/-20pt/{\tau_{E^\ast}}
 & \\
 & & E\ar[rrr]^/-20pt/{\rho}\ar[ddl]_/-20pt/{\tau}
 & & & \sT M \ar[ddl]_{\tau_M}\\
 E^\ast\ar[rrr]^/-20pt/{id}\ar[dr]^{\pi}
 & & & E^\ast\ar[dr]^{\pi} & &  \\
 & M\ar[rrr]^{id}& & & M &
}\label{eqn:pi_diag}
\end{equation}
For Lie algebroids $\Pi$ is a linear Poisson structure on $E^\ast$. It is well recognised in two standard examples: if $E=\T M$, $\Pi$ is a canonical Poisson structure on $\T^\ast M$, whereas if $E=\g$ it is a Lie-Poisson structure on $\g^\ast$.

\paragraph{Algebroids as double-vector-bundle morphisms}
Composing double vector bundle morphism $\iota_\Pi$ with a canonical double vector bundle isomorphism (and antisymplectomorphism) $R_\tau:\T^\ast E\approx\T^\ast E^\ast$  one obtains double vector bundle morphism $\eps_E:\T^\ast E\ra \T E^\ast$. The presence of such a morphism provides an equivalent characterisation of the algebroid structure on $E$. 
\begin{equation}
\xymatrix{
 & \sT^\ast E \ar[rrr]^{\epsilon_E} \ar[dr]^{\pi_E}
 \ar[ddl]_{\sT^\ast\tau}
 & & & \sT E^\ast\ar[dr]^{\sT\pi}\ar[ddl]_/-20pt/{\tau_{E^\ast}}
 & \\
 & & E\ar[rrr]^/-20pt/{\rho}\ar[ddl]_/-20pt/{\tau}
 & & & \sT M \ar[ddl]_{\tau_M}\\
 E^\ast\ar[rrr]^/-20pt/{id}\ar[dr]^{\pi}
 & & & E^\ast\ar[dr]^{\pi} & &  \\
 & M\ar[rrr]^{id}& & & M &
}\label{eqn:eps_diag}
\end{equation}
Since in local coordinates $R_\tau$ is given by
$$ R_\tau(x^a, y^i,p^b, \pi_j)=(x^a, \pi_i, -p_b,y^j),$$
the morphism $\eps_E$ is of the form
\begin{equation}
\label{eqn:eps_coord}
\eps_E(x^a,y^i,p_b,\xi_j) = (x^a, \xi_i, \rho^b_k(x)y^k, c^k_{ij}(x) y^i\xi_k + \sigma^a_j(x) p_a).
\end{equation}

\paragraph{Relation $\bm{\kappa}$}
Since the dual bundles of $\pi_E:\sT^\ast E\ra E$ and $\sT\pi:\sT E^\ast\ra\sT M$ are, respectively,
$\tau_E:\sT E\ra E$ and $\sT\tau:\sT E\ra\sT M$, the dual to $\eps_E$ is a relation $\kappa_E:\sT E\rel\sT E$.
It is a uniquely defined smooth submanifold $\kappa$ in $\sT E\times\sT E$ consisting of pairs $(v,v')$ such that
$\rho(\tau_E(v'))=\sT\tau(v)$ and
$$\< v,\eps_E(v^\ast)>_{\sT\tau}=\< v',v^\ast>_{\tau_E}$$
for any $v^\ast\in\sT^\ast_{\tau_E(v')}E$, where $\< \cdot,\cdot>_{\sT\tau}$ is the canonical pairing
between $\sT E$ and $\sT E^\ast$, and $\< \cdot,\cdot>_{\tau_E}$ is the canonical pairing between $\sT E$
and $\sT^\ast E$. We will write $\kappa:v\rel\, v'$ instead of $(v,v')\in\kappa$. This relation can be put into the
following diagram of "double vector bundle relations"
\begin{equation}
\xymatrix{
 & \sT E  \ar[dr]^{\tau_E}
 \ar[ddl]_{\sT\tau}
 & & & \sT E\ar @{->}[lll]_{\kappa}\ar[dr]^{\sT\tau}\ar[ddl]_/-20pt/{\tau_{E}}
 & \\
 & & E\ar[rrr]^/-20pt/{\rho}\ar[ddl]_/-20pt/{\tau}
 & & & \sT M \ar[ddl]_{\tau_M}\\
 \sT M\ar[dr]^{\tau_M}
 & & & E\ar[dr]^{\tau}\ar[lll]_{\sigma} & &  \\
 & M\ar[rrr]^{id}& & & M &
}\label{eqn:kappa}
\end{equation}
In local coordinates it reads
\begin{equation}
\label{eqn:kappa_coord}\kappa:\left(x^a,\,{Y}^i,\,\rho^b_k(x)y^k,\,\dot{Y}^j\right)\rel\left(x^a,y^i,\,\sigma^b_k(x)Y^k,\,
\dot{Y}^j+c^j_{kl}(x)y^kY^l\right)\,.
\end{equation}

A canonical example of a mapping $\eps_E$ in the case of $E=\sT M$ is given by $\eps_{\T M} = \alpha^{-1}_M$ -- the
inverse to the Tulczyjew isomorphism $\alpha_M:\sT\sT^\ast M\ra\sT^\ast\sT M$ \cite{Tul}. The dual relation is in this case the well-known ``canonical flip'' $\kappa_{M}:\sT\sT M\ra\sT\sT M$. Since $\alpha_M$ is an isomorphism,
$\kappa_{M}$ is a true map, in fact -- an isomorphism of the corresponding two vector bundle structures as well. In local coordinates $(x^a,\dot x^b,\del x^c,\del\dot x^d)$ on $\T\T M$, $\kappa_M$ maps an element $(x^a,\dot x^b,\del x^c,\del\dot x^d)$ to $(x^a,\del x^b,\dot x^c,\del\dot x^d)$.


\section{An abstract approach to variational calculus}\label{sec:var_calc}

\paragraph{Variational problems} A standard data in a variational problem consists of a vector bundle $\tau:F\ra M$ (perhaps with additional structure) and a function $L:F\ra\R$  (\emph{Lagrangian}). Given any smooth path $\gamma:I=[t_0,t_1]\ra F$ we can define the \emph{action}
$$S_L(\gamma):=\int_IL(\gamma(t))\dd t.$$ 
In typical situations we do not want to consider $S_L$ for all $\gamma\in C^\infty(I,F)$, but instead we restrict our attention to a certain subset $\TT\subset C^\infty(I,F)$. Curves in $\TT$ will be called \emph{admissible trajectories}. 
As an example consider the tangent bundle $\tau_{M}:\T M\ra M$. Typically we are interested only in curves $\gamma:I\ra \T M$ which are tangent prolongations of curves in $M$, i.e. $\gamma(t)=\T\left(\tau_{M}\circ\gamma(t)\right)$.  

A \emph{variation} (or a \emph{virtual displacement}) of $\gamma\in\TT$ is a curve $\del \gamma:I\ra\T F$ which projects to $\gamma$ under $\tau_F:\T F\ra F$. The set of all variations will be denoted by $\C$. We will use a symbol $\C_\gamma$ for the set of all variations which project to $\gamma$. Among elements of $\C$ we can distinguish a class $\C^0$ of variations with \emph{vanishing end-points}. Denote by $\C^0_\gamma$ the common part $\C^0\cap\C_\gamma$. Again, sometimes we will like to restrict our attention to \emph{admissible variations} only, that is variations from a certain subset $\WW\subset \C$. We will use a natural notation $\WW^0:=\WW\cap\C^0$, $\WW_\gamma:=\WW\cap\C_\gamma$ and $\WW^0_\gamma:=\WW\cap\C^0_\gamma$.  

For a given $\gamma\in\TT$ and $\del\gamma\in\WW_\gamma$ we can consider the \emph{variation of the action} $S_L$  at $\gamma$, along $\del\gamma$
$$\<\del\gamma,\dd S_L(\gamma)>:=\int_I\<\del\gamma(t),\dd L(\gamma(t))>\dd t.$$ 
We would like to understand a \emph{variational problem} as the study of the critical points of the action $S_L$ restricted to curves in $\TT$ and variations in $\WW$. The variational problem can be hence defined as a quadruple 
\begin{equation}\label{def:var_prob}
\PP:=(F,L,\TT,\WW).
\end{equation} 
Of course, in typical situations the sets $\TT$ and $\WW$ are not arbitrary, but are somehow related to the geometry of the bundle $F$. 

A curve $\gamma \in\TT$ will be called an \emph{extremal} of $\PP$ iff
$$\<\del\gamma,\dd S_L(\gamma)>=0\quad\text{for every $\del\gamma\in\WW_\gamma^0$},$$
i.e. $\gamma$ is a critical point of the action $S_L$ relative to admissible variations with vanishing end-points. The set of all extremals of $\PP$ will be denoted by $\Gamma_\PP\subset\TT$. 

\paragraph{Constraints} The above definition of the variational problem is very general and particularly useful in the context of constrains. We can define $\PP^{'}$ -- a \emph{constrained variational problem} by restricting the set of admissible trajectories and/or admissible variations, i.e. $\PP^{'}=(F,L,\TT^{'},\WW^{'})$, where $\TT^{'}\subset\TT$ and $\WW^{'}\subset\WW$. In this way we follow the ideas of Tulczyjew \cite{Tul1}, namely that constraints should be imposed rather on virtual displacements than on configurations. Usually these restrictions are somehow related to additional geometric structures in the bundle $F$. We will see concrete examples while studying the nonholonomic and vaconomic  constraints associated to a subbundle $D$ of an algebroid $E$. 

\section{Lagrangian and Hamiltonian dynamics on algebroids}\label{sec:lag_ham_alg}

In section \ref{sec:var_calc} we defined an abstract variational problem as a quadruple, consisting of a vector bundle, a Lagrangian function and the sets of admissible trajectories and admissible variations.

Now we will show that if the considered vector bundle is equipped in a general algebroid structure, then the sets of admissible trajectories and variations can be defined in a natural way.

\paragraph{The algebroid variational problem $\bm{\PP_E}$}
Consider a vector bundle $\tau:E\ra M$ with a structure of a general algebroid. Among the curves $\gamma:[t_0,t_1]\ra E$ we distinguish \emph{admissible paths} (or \emph{$E$-paths} shortly), such that the tangent prolongation $\T(x(t))$ of their base curves $x(t):=\tau\circ\gamma(t)$ coincide with their left anchor
\begin{equation}\label{eqn:E_path}
\T(x(t))=\rho(\gamma(t)).
\end{equation}
The set of all such curves will be denoted by $\TT_E$.

Let now $\gamma(t)\in E$ be an $E$-path over $x(t)\in M$ and let $b(t)\in E$ be any other path covering $x(t)$. The tangent prolongation $\T(b(t))$ is a vector field in $\T E$ along $b(t)$. There exists an unique vector field $\del_b\gamma(t)$ in $\T E$ along $\gamma(t)$ (a variation of $\gamma(t)$) which is $\kappa$-related to $\T(b(t))$ 
(cf. \cite{GG} for details). In local coordinates in $\T E$, if $\gamma(t)\sim(x^a(t),y^i(t))$ and $b(t)\sim(x^a(t),b^i(t))$, then
\begin{equation}\label{eqn:var_alg}
\del_b\gamma(t)\sim\left(x^a(t), y^i(t),\sigma^c_i(x)b^i(t), \dot{b}^k(t)+c^k_{ij}(x)y^i(t)b^j(t)\right).
\end{equation}
The set of variations of the above form will be denoted by $\WW_E$. We may say that the variation $\del_b\gamma(t)$ is generated by $b(t)$ -- the \emph{infinitesimal variation}. The subset $\WW^0_E\subset\WW_E$ will consist of variations generated by $b(t)$'s vanishing at the end-points.

Now, for a given Lagrangian function $L:E\ra\R$ we can consider an \emph{algebroid variational problem}
$$\PP_E:=(E,L,\TT_E,\WW_E).$$

Admissible paths and admissible variations have a nice geometric interpretation in the context of a Lie grupoid -- Lie algebroid reduction. Namely, admissible paths are paths in a grupoid reduced to the algebroid and admissible variations are infinitesimal homotopies in a grupoid reduced to the algebroid. For a detailed discussion the reader should confront \cite{CF,GJ}. Let us mention only that the standard variational problems are of the above type. 

\paragraph{Classical variational problem on $\bm{\sT M}$} Consider the standard (unconstrained, with fixed end-points) variational problem of finding the extremal $x(t)\in M$ of the action
$$\int_{t_0}^{t_1}L(x(t),\dot x(t))\dd t,$$
where $L:\T M\ra\R$. This problem is, in fact, a variational problem on a vector bundle $\tau_{M}:\T M\ra M$ with admissible trajectories being just the tangent prolongations of the base curves 
$$\gamma(t)\sim\left(x(t),v(t)\right)\in\T M,\quad \text{where $v(t)=\dot x(t)$}.$$
Admissible variations are defined by means of base homotopies 
$$\del\gamma(t)\sim\left(\del x(t),\del v(t)\right)=\left(\pa_s x(t,s)|_{s=0}, \pa_s \dot x(t,s)|_{s=0}\right),$$
where $x(t,s)$ is the homotopy in $M$ and $x(t,0)=x(t)$. In fact every such variation is generated by a curve $b(t)=\pa_s x(t,s)|_{s=0}\in\T_{x(t)}M$. Indeed, 
$$\left(\del x,\del v(t)\right)=\left(b(t),\kappa_{M}(\dot b(t))\right),$$
where $\kappa_{M}:\T\T M\ra\T\T M$ is the canonical flip which maps $\pa_s(\pa_t x(t,s))$ into $\pa_t\left(\pa_sx(t,s)\right)$. The presence of $\kappa_{M}$ on $\T M$ is often not recognised, being hidden under the obvious symmetry of second derivatives. In local coordinates $(x^a,\dot x^b)$ on $\T M$, if $b(t)\sim(x^a(t),b^c(t))$, then
$$\del x^a(t)=b^a(t)\quad \text{and}\quad \del\dot x^c=\dot b^c(t).$$ 
We recognise equation \eqref{eqn:E_path} and \eqref{eqn:var_alg} for the standard algebroid structure on $\T M$ (in coordinates $(x^a,\dot x^b)$ on $\T M$ we have clearly $\rho^a_c(x)=\sigma^a_c(x)=\del^a_c$ and $c^a_{de}(x)=0$).

\paragraph{Euler-Lagrange equations}
For the algebroid variational problem $\PP_E=(E,L,\TT_E,\WW_E)$ one can derive an analog of Euler-Lagrange equations in a simple manner \cite{GGU,GG}. Take $\gamma\in\TT_E$ and $\del_b\gamma\in(\WW_E)_\gamma$. Let us calculate
\begin{equation}\label{eqn:var_lag_alg}
\<\del_b\gamma,\dd S_L(\gamma)>=\int_{t_0}^{t_1}\<\del_b\gamma(t),\dd L(\gamma(t))>\dd t.
\end{equation}
In local coordinates $(x^a,y^i)$ on $E$ denote $\gamma(t)\sim(x^a(t),y^i(t))$ and $b(t)\sim(x^a(t),b^i(t))$.  Now
\begin{align*}
\<\del_b\gamma(t),\dd L(\gamma(t))>&=\sigma^a_i(x)b^i(t)\frac{\pa L}{\pa x^a}(\gamma(t))+\left(\dot b^i+c^i_{jk}(x)y^j(t)b^k(t)\right)\frac{\pa L}{\pa y^i}(\gamma(t))=\\
&=b^i(t)\left[\sigma^a_i(x)\frac{\pa L}{\pa x^a}(\gamma(t))+c^k_{ji}(x)y^j(t)\frac{\pa L}{\pa y^k}(\gamma(t))-\frac{\dd}{\dd t}\frac{\pa L}{\pa y^i}(\gamma(t))\right]+\frac{\dd}{\dd t}\left[b^i(t)\frac{\pa L}{\pa y^i}(\gamma(t))\right]=\\
&=b^i\, \del L_i(\gamma(t))+\frac{\dd}{\dd t}\left[b^i(t)(\lambda_L)_i(\gamma(t))\right]
,\end{align*}
where $(\lambda_L)_i(x,y):=\frac{\pa L}{\pa y^i}(x,y)$ and $\del L_i(x,y):=\sigma^a_i(x)\frac{\pa L}{\pa x^a}(x,y)+c^k_{ji}(x)y^j\frac{\pa L}{\pa y^k}(x,y)-\frac{\dd}{\dd t}\frac{\pa L}{\pa y^i}(x,y)$.
 
Consequently, 
$$\<\del_b\gamma(t),\dd S_L(\gamma(t))>=\<b(t),\lambda_L(\gamma(t))>|_{t_0}^{t_1}+\int_{t_0}^{t_1}\<b(t),\del L(\gamma(t))>\dd t.$$
Now, $\gamma(t)$ is an extremal of $\PP_E$, if the above expression should be zero for all $b$'s vanishing at the end-points. This gives $\del L(\gamma(t))=0$. In local coordinates $\gamma(t)\sim(x^a(t),y^i(t))$ that reads
\begin{align}\label{eqn:EL}
\frac{\dd}{\dd t}\frac{\pa L}{\pa y^i}(x(t),y(t))&=\sigma^a_i(x(t))\frac{\pa L}{\pa x^a}(x(t),y(t))+c^k_{ji}(x(t))y^j(t)\frac{\pa L}{\pa y^k}(x(t),y(t))
\intertext{which, together with an admissibility condition for $\gamma(t)$} 
\label{eqn:EL1} \dot x^a(t)&=\rho^a_i(x(t))y^i(t),
\end{align} 
generalises the Euler-Lagrange equations. 

The above equations can be also obtained in a purely geometrical way using the diagram \eqref{eqn:eps_diag} (confront \cite{GGU,GG}). Given a Lagrangian function $L:E\ra\R$ one can construct two maps $\lambda_L=\tau_{E^\ast}\circ\dd L:E\ra E^\ast$ and $\Lambda_L=\eps_E\circ\dd L:E\ra \T E^\ast$ as shown on a diagram
 $$\xymatrix{
\T^\ast E  \ar[rr]^{\eps_E}&& \sT E^\ast \ar[d]^{\tau_{E^\ast}} \\
E\ar[rr]^{\lambda_L}\ar[u]^{\dd L} \ar@{-->}[urr]^{\Lambda_L} && E^\ast }.
$$ 
In local coordinates (cf. \eqref{eqn:eps_coord})
\begin{align*}
&\lambda_L(x,y)\sim\left(x^a,\frac{\pa L}{\pa y^i}(x,y)\right),\\
&\Lambda_L(x,y)\sim\left(x^a,\frac{\pa L}{\pa y^i}(x,y),\rho^b_i(x)y^k,c^k_{ij}(x)y^i\frac{\pa L}{\pa y^k}(x,y)+\sigma^a_j(x)\frac{\pa L}{\pa x^a}(x,y)\right).
\end{align*}
Now, Euler-Lagrange equations \eqref{eqn:EL} and  \eqref{eqn:EL1} are simply equivalent to $\Lambda_L(\gamma(t))$ being the tangent prolongation of $\lambda_L(\gamma(t))$
$$\T(\lambda_L(\gamma(t)))=\Lambda_L(\gamma(t)).$$

\paragraph{Vaconomic constraints}
Let now $E$ be an algebroid and $D\subset E$ be a vector subbundle (over a full base $M$ -- compare remark \ref{rem:base}). To define a \emph{vaconomically constrained variational problem} on an algebroid $E$, associated with the bundle $D$, we restrict our attention to admissible paths in $D$ and variations tangent to $D$. That is,
\begin{align*}
&\TT_{D,vac}:=\{\gamma\in \TT_E:\gamma(t)\in D \text{ for every $t\in[t_0,t_1]$}\} \text{ and}\\
&\WW_{D,vac}:=\{\del_b\gamma\in\WW_E: \del_b\gamma(t)\in\T_{\gamma(t)}D \text{ for every $t\in[t_0,t_1]$}\}.
\end{align*}  
In the standard case $E=\T M$, this variations correspond simply to infinitesimal homotopies of paths tangent to the distribution $D$. Note that if $\del_b\gamma\in\WW_{D,vac}$, we have a priori no control of the infinitesimal variation $b(t)$. Since the variations are tangent to $D$, the vaconomic dynamics is determined by the restriction of $L$ to $D$. 

\paragraph{Nonholonomic constraints}
With the same subbundle $D\subset E$ one can associate also a \emph{nonholonomically constrained variational problem}. In this situation we restrict ourselves again to admissible trajectories in $D$
$$\TT_{D,nh}=\TT_{D,vac},$$
but input different restrictions on the set of admissible variations. Namely, according to d'Alembert principle, we consider only these variations $\del_b\gamma\in\WW_E$ which are generated by infinitesimal variations $b(t)$ with values in $D$
$$\WW_{D,nh}=\{\del_b\gamma\in\WW_E:b(t)\in D \text{ for every $t\in[t_0,t_1]$}\}.$$
In this case we have a full control of the set of infinitesimal variations $b(t)$. Note, however, that the variations $\del_b\gamma(t)$ are, in general, no longer tangent to $D$. For this reason the knowledge of $L|_D$ is not sufficient to  study nonhonomically constrained dynamics. 

In cases of both, nonholonomic and vaconomic constraints associated with $D$, the study of action \eqref{eqn:var_lag_alg} with admissible variations restricted to $\WW_{D,nh}$ and $\WW_{D,vac}$, respectively, allows one to derive the algebroid version of the constrained Euler-Lagrange equations on $E$ (confront \cite{GG}) which generalizes the standard constrained dynamics on $\T M$. 

\paragraph{Hamiltonian dynamics} 
As we have already mentioned, an algebroid structure on $E$ is equivalent to a presence of a (2,0)-linear tensor field $\Pi_E$ on the dual bundle $E^\ast$ (linear Poisson bivector field for Lie algebroids). Such a structure allows one to define a natural analog of the Hamiltonian dynamics on $E^\ast$.
 
Namely, for a given function $H:E^\ast\ra \R$ we can  define a Hamiltonian vector field $\X_H$ in $E^\ast$ by a natural formula $\X_H:=\iota_{\dd H}\Pi$. Explicitly, the associated dynamics of $\X_H$ is governed by the following equations
\begin{equation}\label{eqn:ham}
\begin{split}
&\dot \xi_i=\xi_k c^k_{ij}(x)\frac{\pa H}{\pa \xi_j}(x,\xi)-\sigma^a_i(x)\frac{\pa H}{\pa x^a}(x,\xi),\\
&\dot x^a=\rho^a_i(x)\frac{\pa H}{\pa \xi_i}(x,\xi).
\end{split}
\end{equation}

As in the case of the Euler-Lagrange equations it is convenient to describe the construction of $\X_H$ by means of a diagram \eqref{eqn:pi_diag}
 $$\xymatrix{
\T^\ast E^\ast  \ar[rr]^{\iota_\Pi}&& \sT E^\ast \ar[d]^{\tau_{E^\ast}} \\
E^\ast\ar[rr]^{\id}\ar[u]^{\dd H} \ar@{-->}[urr]^{\X_H} && E^\ast }.
$$

In the next chapter we will derive Hamilton equations \eqref{eqn:ham} by means of a variational principle in the sense introduced in previous sections. What is more, the general approach to constraints presented before allows one to derive Hamiltonian  equations on $E^\ast$ associated with constraints in the configuration bundle $E$.

\section{Variational description of Hamiltonian dynamics}\label{sec:results}

\paragraph{Geometric setting}
Let us introduce a few definitions. We will develop a variational calculus on a bundle $F:=E\oplus_ME^\ast\ra M$. The tangent bundle $\T F$ of $F$ can be canonically identified with $\T F=\T(E\oplus_ME^\ast)\approx\T E\oplus_{\T M}\T E^\ast$. Both $F$ and $\T F$ admit natural projections which will be denoted by $\p_E:F\ra E$, $\p_{E^\ast}:F\ra E^\ast$, $\p_{\T E}:\T F\ra\T E$ and $\p_{\T E^\ast}:\T F\ra \T E^\ast$ respectively. Finally, we have natural coordinates $(x^a,y^i,\xi_j)$ on $F$ and $(x^a,y^i,\xi_j,\dot x^b,\dot y^k,\dot\xi_l)$ on $\T F$.

For $F$ as above we can construct sets of admissible trajectories $\TT_F$ and admissible variations $\WW_F$ as follows 
\begin{align*}
\TT_F&:=\left\{\Gamma:[t_0,t_1]\ra F:\p_E\circ\Gamma\in\TT_E\right\},\\
\WW_F&:=\left\{\del\Gamma:[t_0,t_1]\ra\T F:\p_{\T E}\circ\del\Gamma\in\WW_E\right\}.
\end{align*}
That is, admissibility is determined by the $E$ or $\T E$ part of a path or a variation only. We will distinguish the class $\WW^0_F$ consisting of those variations $\del F$ which project to $\WW^0_E$ under $\p_{\T E}$. 

For a given function $H:E^\ast\ra\R$ (Hamiltonian) we may define a new (Lagrangian) function $\LL_H:F\ra\R$ by the formula 
$$\LL_H(X):=\<\p_EX,\p_{E^\ast}X>-H(\p_{E^\ast}X).$$
It defines an action on the space of admissible trajectories in $F$  
$$S_{\LL_H}(\Gamma)=\int_{t_0}^{t_1}\<\p_E\circ\Gamma(t),\p_{E^\ast}\circ\Gamma(t)>-H(\p_{E^\ast}\circ\Gamma(t))\dd t.$$ 
Now we have all ingredients (namely $F$, $\TT_F$, $\WW_F$ and $\LL_H$) necessary to develop variational calculus on $F$ in the sense of section \ref{sec:var_calc}.

\begin{thm}\label{thm:var_ham}
The trajectories $\xi(t)\in E^\ast$ of the Hamiltonian vector field $\X_H$ are precisely the $E^\ast$-projections of the extremals $\Gamma$ of the action $S_{\LL_H}(\cdot)$ with respect to variations in $\WW^0_F$.   

For critical $\Gamma(t)$ which projects to $\xi(t)\in E^\ast$, the $E$-projection is simply $\p_E(\Gamma(t))=\T^\ast\pi\circ\dd H(\xi(t))\in E$. 
\end{thm}

\begin{proof}
For simplicity we may assume that the extremal $\Gamma(t)\sim(x^a(t),y^i(t),\xi_j(t))$ lies in a single coordinate chart of $E\oplus_ME^\ast$. 

Every admissible variation $\del \Gamma\in \WW_F$ along $\Gamma$ has the form
$$\del\Gamma(t)=\sigma^a_i(x)f^i(t)\pa_{x^a}+\left(\dot f^k+c^k_{ij}(x)y^i(t)f^j(t)\right)\pa_{y^k}+h_k(t)\pa_{\xi_k}.$$
Now, 
\begin{align*}
\<\del\Gamma,\dd S_{\LL_H}>&=\int_{t_0}^{t_1}\<\p_{\T E}\circ\del\Gamma(t),\p_{\T E^\ast}\circ\del\Gamma(t)>_{\T\tau}-\<\p_{\T E^\ast}\circ\del\Gamma(t),\dd H(\p_{E^\ast}\circ\Gamma(t))>\dd t=\\
&=\int_{t_0}^{t_1}(\dot f^k+c^k_{ij}(x)y^i(t)f^j(t))\xi_k(t)+y^k(t)h_k(t)-\sigma^a_i(x)f^i(t)\frac{\pa H}{\pa x^a}-h_k(t)\frac{\pa H}{\pa\xi_k}\dd t.
\end{align*}
Integrating $\dot f^k\xi_k$ by parts we will get 
$$\<\del\Gamma,\dd S_{\LL_H}>=f^k\xi_k\Big|_{t_0}^{t_1}+\int_{t_0}^{t_1}f^j\left(-\dot\xi_j-\sigma^a_j(x)\frac{\pa H}{\pa x^a}+\xi_k c^k_{ij}(x)y^i\right)+h_k(t)\left(y^k-\frac{\pa H}{\pa \xi_k}\right)\dd t.$$
If $\del\Gamma\in \WW_F^0$, then $f^k(t)$ vanishes at $t_0$ and $t_1$, hence 
\begin{equation}\label{eqn:var}
\<\del\Gamma,\dd S_{\LL_H}>=\int_{t_0}^{t_1}f^j\left(-\dot\xi_j-\sigma^a_j(x)\frac{\pa H}{\pa x^a}+\xi_k c^k_{ij}(x)y^i\right)+h_k(t)\left(y^k-\frac{\pa H}{\pa \xi_k}\right)\dd t.\end{equation}
The last expression vanishes for every such $\del\Gamma$ if and only if
\begin{align*}
\dot\xi_j&=\xi_k c^k_{ij}(x)y^i-\sigma^a_j(x)\frac{\pa H}{\pa x^a}(x,\xi)
\intertext{and}
y^k&=\frac{\pa H}{\pa \xi_k}(x,\xi).
\end{align*}
Consequently, since by assumptions $\p_E\circ\Gamma(t)\sim(x^a(t),y^i(t))$ is admissible, that is $\dot x^a=\rho^a_i(x)y^i$, we get
\begin{align*}
\dot x^a=&\rho^a_k(x)\frac{\pa H}{\pa \xi_k}(x,\xi),\\
\dot\xi_j=&\xi_k c^k_{ij}(x)\frac{\pa H}{\pa \xi_i}(x,\xi)-\sigma^a_j(x)\frac{\pa H}{\pa x^a}(x,\xi),
\end{align*}
in perfect agreement with \eqref{eqn:ham}. In other words, $\p_{E^\ast}\circ\Gamma(t)\sim(x^a(t),\xi_j(t))$ is a Hamiltonian trajectory of $\X_H$. What is more, $\p_{E^\ast}\circ\Gamma(t)\sim(x^a(t),\frac{\pa H}{\pa\xi_i}(x(t),\xi(t)))$, which is precisely the coordinate form of $\T^\ast\pi\circ\dd H(\xi(t))$. 
\end{proof}

\section{Hamiltonian dynamics with constraints}\label{sec:constraints}

As we have seen, Hamiltonian dynamics can be described by means of a variational principle on $F=E\oplus_ME^\ast$, provided that we impose simple condition on admissible variations $\p_{\T E}\circ\del\Gamma\in\WW_E$ and trajectories $\p_{E}\circ\Gamma\in\TT_E$. It is a natural idea to define constrained admissible variations and trajectories, associated with a subbundle $D\subset E$, in a similar manner;
\begin{align*}
\TT_{F,D,vac}&:=\{\Gamma\in\TT_F:\p_{E}\circ\Gamma\in\TT_{D,vac}\},\\
\TT_{F,D,nh}&:=\{\Gamma\in\TT_F:\p_{E}\circ\Gamma\in\TT_{D,nh}\},\\
\WW_{F,D,vac}&:=\{\del\Gamma\in\WW_F:\p_{T E}\circ\del\Gamma\in\WW_{D,vac}\},\\
\WW_{F,D,nh}&:=\{\del\Gamma\in\WW_F:\p_{T E}\circ\del\Gamma\in\WW_{D,nh}\},
\end{align*}
and to derive a constrained version of the Hamiltonian dynamics using the variational approach. It seems also possible to consider constraints in the phase part $E^\ast$ of $F$. However, since $E^\ast$ does not posses  a priori any natural structure to generate constrained variations (such as an algebroid structure), these could be only configuration constraints.

\paragraph{Nonholonomic constraints}
In this case we can calculate the associated equations of motion in a manner similar to the proof of theorem \ref{thm:var_ham}. Indeed, we can derive formula \eqref{eqn:var} again, yet now, since $\del\Gamma\in\WW_{F,D,nh}$ the curve $f^i(t)$ belongs to $D$, whereas $h_k(t)$ can be arbitrary. Consequently, we get that 
$$\dot\xi_j-\xi_k c^k_{ij}(x)y^i+\sigma^a_j(x)\frac{\pa H}{\pa x^a}(x,\xi)$$
annihilates $D$. In other words, if $D$ is given locally via implicit equations $\Phi^s(x,y)=0$, there exists functions $\mu_s(t)$ such that
$$\dot\xi_j=\xi_k c^k_{ij}(x)y^i-\sigma^a_j(x)\frac{\pa H}{\pa x^a}(x,\xi)+\mu_s(t)\frac{\pa\Phi^s}{\pa y^j}\left(x,\frac{\pa H}{\pa \xi}\right).$$
In addition, we have the condition $y^k=\frac{\pa H}{\pa \xi_k}(x,\xi)$ and the fact that the $E$-projection of the extremal $\p_E\circ\Gamma\sim(x^a,y^i)=(x^a,\frac{\pa H}{\pa\xi_i})$ should belong to $\TT_{D,nh}$. As a consequence we get the following set of differential equations
\begin{align}
\label{eqn:ham_con_nh1}&\dot\xi_j=\xi_k c^k_{ij}(x)\frac{\pa H}{\pa \xi_i}(x,\xi)-\sigma^a_j(x)\frac{\pa H}{\pa x^a}(x,\xi)+\mu_s(t)\frac{\pa\Phi^s}{\pa y^j}\left(x,\frac{\pa H}{\pa \xi}\right),\\
\label{eqn:ham_con_nh2}&\dot x^a=\rho^a_k(x)\frac{\pa H}{\pa \xi_k}(x,\xi),\\
\label{eqn:cond_nh} &\Phi^s(x,\frac{\pa H}{\pa\xi})=0.
\end{align}

The same equations (yet for the standard Lie algebroid structure on $\T M$) have been obtained for instance in \cite{SM}. In the derivations the authors assumed that the Hamiltonian comes from a regular Lagrangian via the Legendre map. In our approach no such an assumption is necessary -- we may simply start from an arbitrary Hamiltonian.
Equations \eqref{eqn:ham_con_nh1}-\eqref{eqn:cond_nh}, again for Hamiltonian given via a Legendre map of a regular Lagrangian, can be also obtained from the nonholonomically constrained Euler-Lagrange equations on algebroids - compare \cite{GG} (equations (5.10) and (5.11)).

\paragraph{Geometric approach to nonholonomic constraints}

Equations \eqref{eqn:ham_con_nh1}--\eqref{eqn:cond_nh} can be described in a purely geometric way. The construction is based on the diagram \eqref{eqn:pi_diag}. Observe first that the inclusion $i:D\ra E$ gives rise to a dual map $i^\ast:E^\ast\ra D^\ast$ and consequently to a tangent map $\T i^\ast:\T E^\ast\ra \T D^\ast$ over $i^\ast$. Composing \eqref{eqn:pi_diag} with the later we will obtain the following diagram
$$\xymatrix{
\T^\ast E^\ast  \ar[rr]^{\iota_\Pi}&& \sT E^\ast \ar[rr]^{\T i^\ast} && \T D^\ast\ar[d]^{\tau_{D^\ast}}\\
E^\ast\ar[rr]^{\id}\ar[u]^{\dd H} \ar@{-->}[urrrr]^{\X^D_H} && E^\ast \ar[rr]^{i^\ast}&& D^\ast},
$$ 
where $\X^D_H:=\T i^\ast\circ \X_H$ is a map from $E^\ast$ to $\T D^\ast$.

\begin{thm} The curve $\gamma(t)\sim(x^a(t),\xi_i(t))\in E^\ast$ satisfies the nonholonomic Hamilton equations \eqref{eqn:ham_con_nh1}--\eqref{eqn:ham_con_nh2} if and only if $\X^D_H(\gamma(t))\in \T D^\ast$ is a tangent prolongation of a curve $i^\ast\gamma(t)\in D^\ast$. Condition \eqref{eqn:cond_nh} means simply that $\T ^\ast\pi(\dd H)$ belongs to $D$. 
\end{thm}
\begin{proof} For $X\in \T E^\ast$ consider a vector
$$\wh\del H(X):=\X_H(\tau_{E^\ast}(X))-X\in \T E^\ast.$$
In local coordinates if $X\sim(x^a,\xi_i,\dot x^b,\dot \xi_j)$ then
$$\wh\del H(X)\sim\left(x^a,\xi_i,\rho^b_i(x)\frac{\pa H}{\pa\xi_i}(x,\xi)-\dot x^b,\xi_k c^k_{ij}(x)\frac{\pa H}{\pa \xi_i}(x,\xi)-\sigma^a_j(x)\frac{\pa H}{\pa x^a}(x,\xi)-\dot\xi_j\right).$$
Now $X=\T(\gamma(t))$ satisfies \eqref{eqn:ham_con_nh2} if and only if $\wh\del H(X)$ is a vertical vector, whereas \eqref{eqn:ham_con_nh1} means that its vertical part $V(\wh\del H(X))=:\del H(X)\sim\left(x^a,\xi_k c^k_{ij}(x)\frac{\pa H}{\pa \xi_i}(x,\xi)-\sigma^a_j(x)\frac{\pa H}{\pa x^a}(x,\xi)-\dot\xi_j\right)$ belongs to the annihilator $D^0\subset E^\ast$. Observe that $\del H(X)\in D^0$ if and only if $i^\ast \del H(X)=0$. Finally note that since $i^\ast: E^\ast\ra D^\ast$ is the identity on the base, a vector $Y\in \T E^\ast$ is vertical if and only if $\T i^\ast Y$ is vertical. Moreover, for a vertical vector $Y\in VE^\ast$ we have $i^\ast V(Y)=V(\T i^\ast Y)$. 

Collecting all the above fact we conclude that \eqref{eqn:ham_con_nh1} and \eqref{eqn:ham_con_nh2} is equivalent to $\wh\del H(T(\gamma(t)))\in V E^\ast$ and $i^\ast V(\wh\del H(T(\gamma(t))))=0$ which is, in turn, equivalent to $\T i^\ast\left(\wh\del H(T(\gamma(t)))\right)=0$. This, by definition of $\wh\del H$ and $\X^D_H$, means that 
$$\X^D_H(\gamma(t))=\T i^\ast \X_H(\gamma(t))=\T(i^\ast\gamma(t)).$$
\end{proof}

\paragraph{ Nonholonomic reduction}

Considerations from the last paragraph suggest that the dual bundle $D^\ast$ plays the role of a phase space for the nonholonomic dynamics. This in turn leads to a natural question is it possible to formulate the dynamics purely in therms of the bundle $D^\ast$. Such a procedure is known as a \emph{nonholonomic reduction}. 

This problem for $D\subset\T M$ was considered by van der Shaft and Maschke \cite{SM}. The authors describe the nonholonomically constrained Hamiltonian dynamics on $\T^\ast M$ (for regular Hamiltonians) as an unconstrained Hamiltonian dynamics on $D^\ast$ obtained using a ''Poisson'' bracket on $D^\ast$ which needs not to satisfy the Jacobi identity. Using our terminology the structure they consider is simply a bi-tensor field on $D^\ast$ which corresponds to a general algebroid structure on $D$. 

The above results seems very interesting as they allow to cover the constrained and unconstrained dynamics within a universal algebroid formalism. Unfortunately the nonholonomic reduction of \cite{SM} has no geometric meaning. Strictly speaking the authors are unaware of using the decomposition $\T M=D\oplus D^{'}$ which depends on the choice of local coordinates. 

In some cases, however, where a natural decomposition $E=D\oplus D^{'}$ can be given, the nonholonomic reduction can be carried out. Below we will state a result of this kind, which can be understood as a Hamiltonian analog of the Lagrangian nonholonomic reduction described in \cite{GLMM}. We do not give a proof (which is quite simple) since this paragraph was intended to be informative only.

\begin{thm}
Consider a nonholonomically constrained algebroid variational problem associated with a subbundle $D\subset E$ of an algebroid $(E,[\cdot,\cdot],\rho,\sigma)$ for a Hamiltonian $H:E^\ast\ra\R$. Assume that $H$ is of mechanical type, i.e. $H(x,\xi)=\frac 12g^{-1}(\xi,\xi)+V(x)$, where $g$ is a non-degenerate symmetric bilinear form on $E$ and $V:M\ra\R$ is a smooth base function.

Let $E=D\oplus D^\perp$ and $E^\ast=D^\ast\oplus (D^\perp)^\ast$ be natural decompositions and $P:E\ra D$ a natural projection induced by $g$.

The nonholonomically constrained Hamiltonian trajectories on $E^\ast$ are the same as unconstrained Hamiltonian trajectories on $D^\ast$ derived for a Hamiltonian $H|_{D^\ast\oplus\{0\}}$ and an algebroid structure $(D,[\cdot,\cdot]_D,\rho_D,\sigma_D)$ given by the bracket $[\cdot,\cdot]_D:=P[\cdot,\cdot]$, and anchors $\rho_D:=\rho|_D$ and $\sigma_D:=\sigma|_D$. 
\end{thm}

\paragraph{Vaconomic constraints}
The derivation of the constrained equations of motion is much more delicate in this case. Repeating the derivations from the proof of theorem \ref{thm:var_ham} we can again obtain \eqref{eqn:var}. We can easily conclude that $y^k=\frac{\pa H}{\pa \xi_k}(x,\xi)$, since $h_k(t)$ can be arbitrary. Now, however, we have no clear condition for $f^k(t)$ apart from $\p_{\T E}\circ\del\Gamma\in \T D$, so it is difficult to derive the second equation.

Instead, we can give a simple sufficient condition for $\Gamma$ being the extremal in this case. Observe that if $\Phi:\R\times E\ra\R$ is a function vanishing on $D$, then, since $\p_{\T E}\circ\del\Gamma$ is tangent to $D$, we have $\<\del\Gamma(t), \dd\Phi(t,\p_E\circ\Gamma(t))>=0$. Now, if for any such $\Phi$ the curve $\Gamma$ is an extremal of the unconstrained variational problem on $F$ for the new Lagrangian function $\wt\LL_H(X):=\LL_H(X)+\Phi(t,\p_EX)$, then for $\del\Gamma\in\WW^0_{F,D,vac}$ we have
$$0=\<\del\Gamma(t),\dd\wt\LL_H(\Gamma)(t)>=\<\del\Gamma(t),\dd\LL_H(\Gamma(t))+\dd\Phi(t,\p_E\circ \Gamma(t))>=\<\del\Gamma(t),\dd\LL_H(\Gamma(t))>.$$
In consequence, $\Gamma$ is an extremal of the vaconomically constrained problem. If $D$ is given by the set of implicit equations $\Phi^s(x,y)=0$, we may take $\Phi(t,x,y)$ of the form $\mu_s(t)\Phi^s(x,y)$. The equations of motion in this case read

\begin{align}
\label{eqn:ham_con_vac1}&\frac{\dd}{\dd t}\left(\xi_j+\mu_s(t)\frac{\pa\Phi^s}{\pa y^j}\right)=\left(\xi_j+\mu_s(t)\frac{\pa\Phi^s}{\pa y^k}\right) c^k_{ij}(x)\frac{\pa H}{\pa \xi_i}-\sigma^a_j(x)\left(\frac{\pa H}{\pa x^a}-\mu_s(t)\frac{\pa \Phi^s}{\pa x^a}\right),\\
\label{eqn:ham_con_vac2}&\dot x^a=\rho^a_k(x)\frac{\pa H}{\pa \xi_k}(x,\xi),
\intertext{with the additional set of conditions:}
\label{eqn:cond_vac}
&\Phi^s(x,\frac{\pa H}{\pa\xi})=0.
\end{align} 

To link the above equations to known results, note that the vaconomicaly constrained Euler-Lagrange equations on general algebroids have been derived in \cite{GG} (equations (5.3) and (5.4)). In case that the Hamiltonian $H$ comes from a regular Lagrangian $L$ via a Legendre map, equations \eqref{eqn:ham_con_vac1}--\eqref{eqn:cond_vac} can be eaisily derived from these equations.

\begin{remark}\label{rem:base}
In this paper we introduced the constraints (both nonholonomic and vaconomic) by means of a subbundle $D\subset E$ over a full base $M$. However, most of the results would remain valid also if $D$ was a subbundle over a submanifold $M_D\subset M$ if we assumed admissibility conditions $\rho(D)\subset \T M_D$ and $\sigma(D)\subset\T M_D$. We decided to work in a less general situation for notation simplicity. 
\end{remark}

\end{document}